\documentclass[]{article}
\usepackage[margin=0.75in]{geometry}
\usepackage{authblk}
\newcommand{\arxiv}[1]{#1}
\newcommand{\journal}[1]{}

\journal{\documentclass[english,preprint,JIP]{ipsj}
\bibliographystyle{ipsjunsrt-e}
}

\usepackage{graphicx}
\usepackage{latexsym}

\usepackage[square,sort,comma,numbers]{natbib}

\usepackage{amssymb,amsmath,amsthm}

\usepackage{hyperref}
\usepackage{enumitem}

\usepackage{booktabs}
\usepackage{graphicx}
\usepackage{wrapfig}
\usepackage{hyperref}
\usepackage{tabularx}

\usepackage{listings}

\newcommand{\defn}[1]{\textbf{\textit{\boldmath #1}}}

\allowdisplaybreaks

\usepackage{xcolor}
\usepackage{soul}

{\makeatletter
 \gdef\xxxmark{%
   \expandafter\ifx\csname @mpargs\endcsname\relax %
     \expandafter\ifx\csname @captype\endcsname\relax %
       \marginpar{xxx}%
     \else
       xxx %
     \fi
   \else
     xxx %
   \fi}
 \gdef\xxx{\@ifnextchar[\xxx@lab\xxx@nolab}
 \long\gdef\xxx@lab[#1]#2{\textbf{[\xxxmark #2 ---{\sc #1}]}}
 \long\gdef\xxx@nolab#1{\textbf{[\xxxmark #1]}}
}

\def\Underline{\setbox0\hbox\bgroup\let\\\endUnderline}
\def\endUnderline{\vphantom{y}\egroup\smash{\underline{\box0}}\\}
\def\|{\verb|}

\newtheorem{problem}{Problem}

\let\epsilon=\varepsilon

\newtheorem{theorem}{Theorem}
\newtheorem{corollary}{Corollary}
\newtheorem{lemma}{Lemma}

\graphicspath{{figures}}

\usepackage[varg]{txfonts}%
\makeatletter%
\input{ot1txtt.fd}
\makeatother%

\begin{document}

\title{Folding One Polyhedral Metric Graph into Another}

\journal{
\affiliate{mit}{MIT Computer Science and Artificial Intelligence Laboratory,
      32 Vassar St., Cambridge, MA 02139, USA}
\affiliate{meiji}{Meiji Institute for Advanced Study of Mathematical Sciences, Meiji University, Nakano, Tokyo 164-8525, Japan}
\author{Lily Chung}{mit}[lkdc@mit.edu]
\author{Erik D. Demaine}{mit}[edemaine@mit.edu]
\author{Martin L. Demaine}{mit}[mdemaine@mit.edu]
\author{Markus Hecher}{mit}[hecher@mit.edu]
\author{Rebecca Lin}{mit}[ryelin@mit.edu]
\author{Jayson Lynch}{mit}[jaysonl@mit.edu]
\author{Chie Nara}{meiji}
[cnara@jeans.ocn.ne.jp]}

\arxiv{%
\author[a]{Lily Chung}
\author[a]{Erik D. Demaine}
\author[a]{Martin L. Demaine}
\author[a]{Markus Hecher}
\author[a]{Rebecca Lin}
\author[a]{Jayson Lynch}
\author[b]{Chie Nara}
\affil[a]{MIT Computer Science and Artificial Intelligence Laboratory,
      32 Vassar St., Cambridge, MA 02139, USA, {\{lkdc,edemaine,mdemaine,hecher,ryelin,jaysonl\}@mit.edu}}
\affil[b]{Meiji Institute for Advanced Study of Mathematical Sciences, Meiji University, Nakano, Tokyo 164-8525, Japan,
{cnara@jeans.ocn.ne.jp}
}%
\maketitle}

\begin{abstract}
We analyze the problem of folding one polyhedron, viewed as a metric graph of its edges, into the shape of another, similar to 1D origami. We find such foldings between all pairs of Platonic solids and prove corresponding lower bounds, establishing the optimal scale factor when restricted to integers.
Further,
	we establish that our folding problem is also NP-hard, even if the source graph is a tree. It turns out that the problem is hard to approximate,
	as we obtain NP-hardness even for determining the existence of a scale factor $1.5-\epsilon$.
Finally, we prove that, in general, the optimal scale factor has to be rational. This insight then immediately results in NP membership. 
In turn, verifying whether a given scale factor is indeed the smallest possible, requires two independent calls to an NP oracle, rendering the problem DP-complete.
\end{abstract}

\journal{\maketitle}

\section{Introduction}

\noindent
Viewing a polyhedron as a \defn{metric graph}
(graph with specified edge lengths),
when can we \defn{fold} it into another polyhedron,
in the sense of 1D origami \cite{demaine2007geometric}
where lengths must be preserved and we view multiple overlapping layers as one?
More formally:

\medskip
\begin{problem}[IsoCovering]\label{prob:main}
    Given two metric spaces $A$ and $B$, find an \defn{isometric covering} of $B$ by $A$, that is, a surjective map $m: A \to B$ such that, for every path $p$ in $A$, the arc length of $p$ in $A$ equals the arc length of $m(p)$ in $B$.
\end{problem}
\smallskip

\begin{figure}[h]%
  \centering
\includegraphics[width=\linewidth]{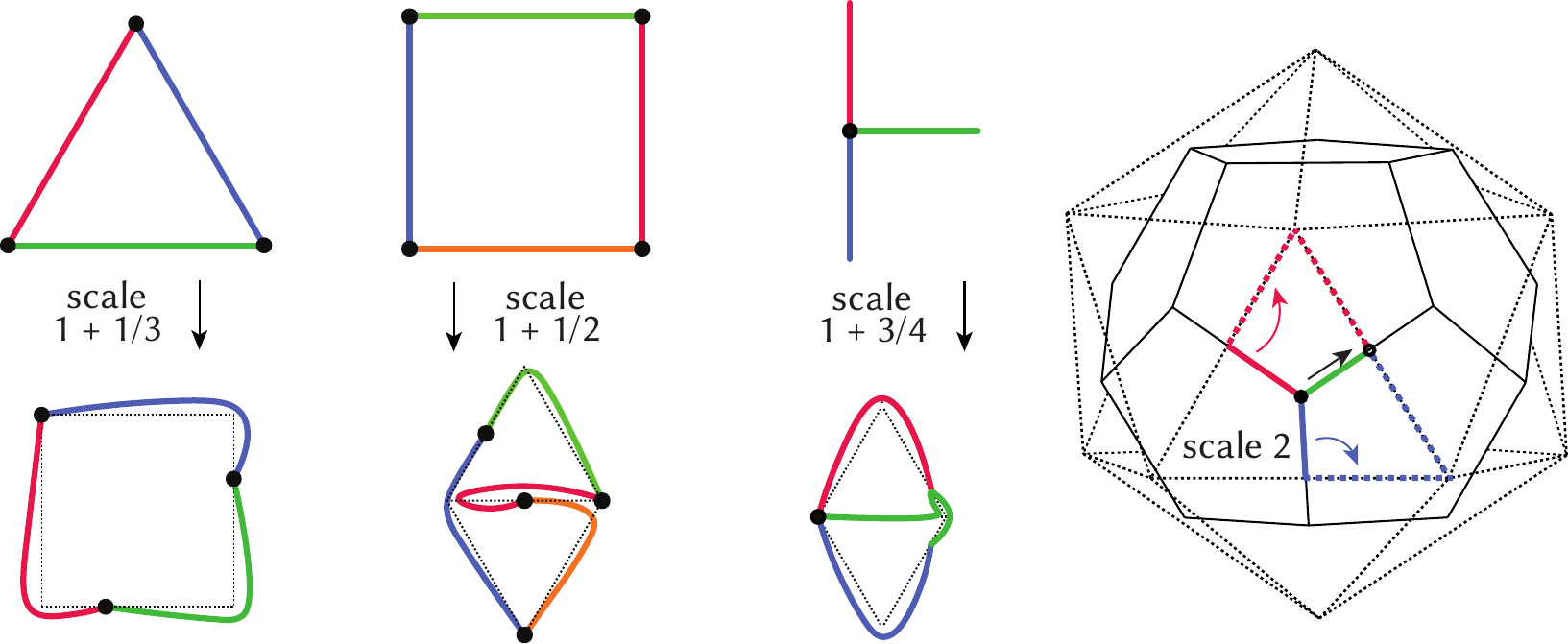}
  \caption{(Left): Simple example foldings, where source graphs are routed onto target graphs. (Right): Folding icosahedron into a dodecahedron via inscription of the target geometry into the source.}
  \label{fig:example}
  \vspace{-10pt}
\end{figure} 

\arxiv{\medskip}
Figure~\ref{fig:example} shows simple examples.
To ensure that this is always possible, we can scale the lengths in $A$ by a constant \defn{scale factor} $\alpha$ (in Figure~\ref{fig:example}, scale factors $\alpha$ are given for unit edge lengths). %
The required scale factor between source and target shape
is also the primary concern in computational origami design
\cite{demaine2007geometric}.

We aim to compute the minimum possible scale factor.
This leads to the decision problem of verifying whether a given scale factor is possible as well as the following optimization problem. %

\medskip
\begin{problem}[Optimization]
    Given graphs $A$ and $B$, minimize $\alpha\in\mathbb{R}$ such that $B$  scaled by $\alpha$ has an isometric covering from~$A$.
\end{problem}
\smallskip

Some well-known problems can be cast into this framework.
When the source graph $A$ is a path or a cycle, this problem is
exactly the \defn{Chinese postman tour} problem:
what is the shortest cycle that visits every edge at least once?
This problem, and thus optimal scale factors in this case,
can be computed in polynomial time~\cite{EdmondsJohnson73}.
Notably, folding cycles like this are commonly studied in the context of
1D origami~\cite[Chapters 12]{DemaineRourke07}.

\subsection{Our Results}

\noindent 
First, we consider foldings of every Platonic solid into every other,
proving many upper bounds (Section~\ref{sec:upper})
and lower bounds (Section~\ref{sec:lower}) on the optimal scale factor~$\alpha$.
Table~\ref{tab:UpperBounds} summarizes these results,
which are tight if we artificially restrict to integral scale factor.

Next, we analyze the complexity of the general problem.
We prove that IsoCovering is NP-hard and that the optimization problem is hard to approximate within some constant factor (Section~\ref{sec:hard}).
We prove that the optimal scale factor is always rational
(Section~\ref{sec:rational}),
which immediately gives rise to NP membership.
We also show DP-completeness of verifying that a given scale factor
is the optimal (Section~\ref{sec:consequences}).
We conclude with open problems in Section~\ref{sec:open}.

\begin{table*}[t]
\caption{Our integer-tight lower and upper bounds, given as intervals, on the minimum scale factor for folding one Platonic solid (row) into another (column), both with unit edge lengths. $*$ The lower bound from icosahedron to cube is due to topological arguments, see also Theorem~\ref{thm:icosacube}.}
\begin{tabularx}{\textwidth}{l@{\hspace{.25em}}|@{\hspace{.45em}}X@{\hspace{.8em}}X@{\hspace{.8em}}X@{\hspace{.8em}}X@{\hspace{.8em}}X}

\toprule 
        $\Rsh$     & Tetrahedron & Cube                    & Octahedron & Dodecahedron & Icosahedron \\ \midrule 
Tetrahedron & [1, 1]  & [2+1/3, 2+5/6] & (2, 2+1/2]  & [6+1/3, 6+5/6]  & (5+2/3, 5+7/8]       \\ %
Cube& (1/2, 5/6] & [1, 1] & (1, 1+1/2]    & [3, 3]   & (2+2/3, 3]   \\ %
Octahedron   & (1/2, 1] & (1+1/12, 1+1/2] & [1, 1]   & (3+1/12, 4]  & (2+1/2, 3]  \\ %
Dodecahedron & (1/5, 3/5] & (2/5, 4/5] & (2/5, 3/4] & [1, 1] & (1, 1+1/3]  \\ %
Icosahedron  & (1/5, 1]  & (1*, 1+1/3]  & (2/5, 1]  & (1+2/15, 2]  & [1, 1] \\ \bottomrule                             
\end{tabularx}
\label{tab:UpperBounds}
\end{table*}

\section{Platonic Upper Bounds (Foldings)}\label{sec:upper}
\label{sec:solutions}

\noindent To achieve the upper bounds reported in Table~\ref{tab:UpperBounds}, we developed initial solutions by observing inscriptions of one polyhedron in another (Figure~\ref{fig:example} (right)), and then further optimized these solutions through manual rerouting and automated search. For the latter, we combined two brute-force paradigms---logic programming and integer linear programming---along with local improvement techniques. 

\begin{figure*}[th]\arxiv{\vspace{-3.35em}}
\includegraphics[width=0.97\textwidth]{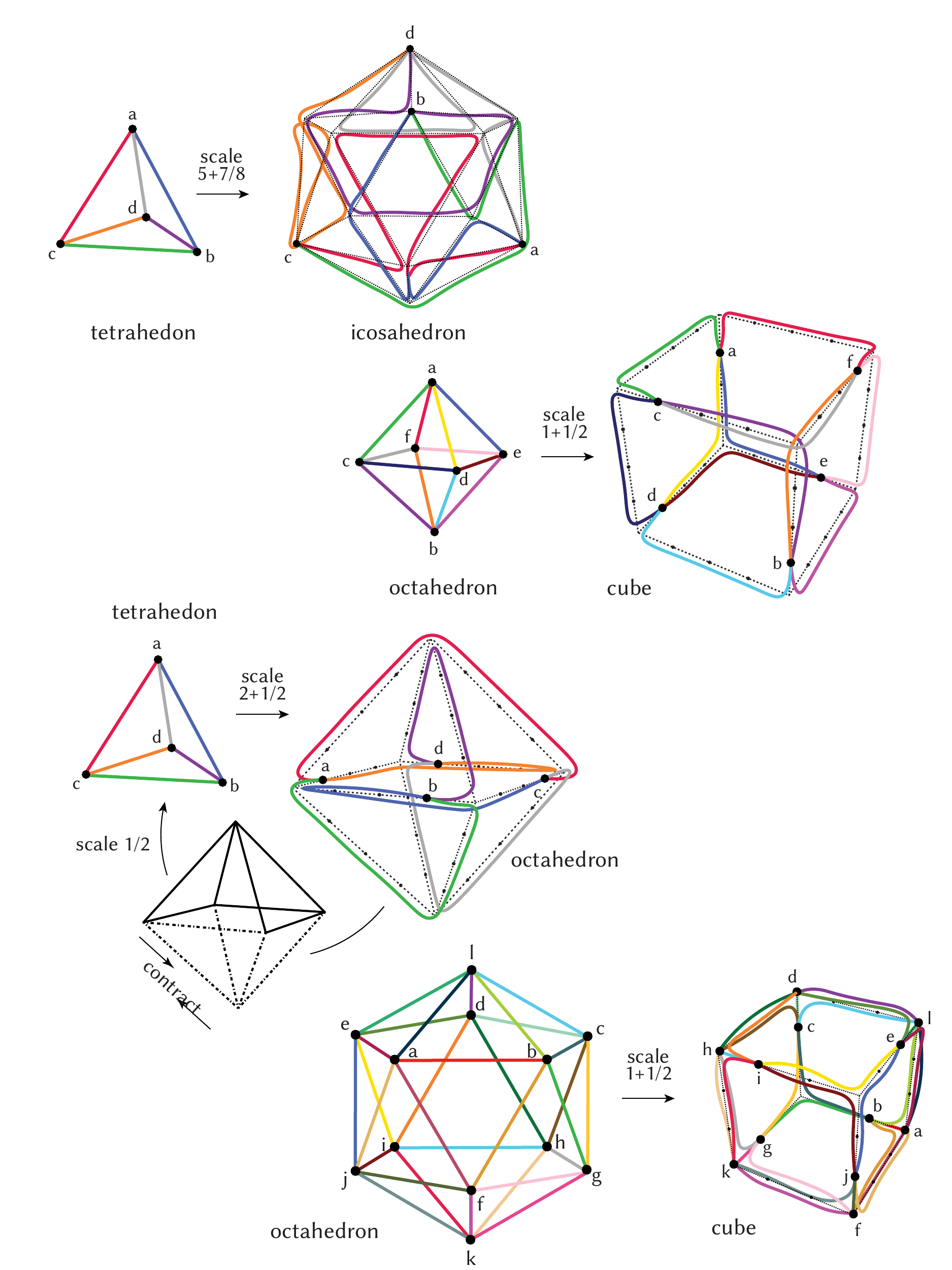}\arxiv{\vspace{-1em}}
\caption{Routing edges via subdivisions of (a) a tetrahedron to an icosahedron with $\alpha=47/8$, (b) an octahedron to a cube with $\alpha=3/2$, (c) a tetrahedron to an octahedron with $\alpha=5/2$ (and back with scale $\alpha=1$), and (d) an octahedron to a cube with $\alpha=3/2$. Note that mappings are schematic and, therefore, not drawn to scale.}\label{fig:solutions}
\end{figure*}

\subsection{Automated Upper Bounds}

\noindent For obtaining upper bounds, we use two different computing techniques. First, we will discuss an approach based on logic programming, which inherits concepts from propositional satisfiability (SAT) with direct language supports for transitive closure (reachability). %
Then, we briefly present an integer linear program.
At the end of this section, we describe our approach of combining both technologies in the form of a local improvement.

\medskip
\subsubsection{Integer Scales}\label{sec:lp}
\noindent 
Modern SAT-based search engines~\cite{FichteEtAl23} provide an efficient tool for searching for solutions to combinatorial problems. The advantage of these solvers is that they are based on non-chronological backtracking, thereby using sophisticated heuristics to efficiently traverse and cut the search space. With the help of sophisticated conflict analysis, the idea is to learn from unsuccessful branches of the search tree, with the goal of avoiding (short-cutting) similar conflicts in the future. %

We design an encoding to utilize a SAT-based logic programming solver, called clingo~\cite{GebserEtAl19}, which is particularly focused on efficiently evaluating transitive closure properties by lazily adding constraints during solving. %
While we do not expect to fully solve optimization problems in general, in practice we can use our encoding to efficiently compute solutions of decent quality.

However, as this method is SAT-based it translates to Boolean variables. Consequently, we cannot expect arbitrarily large precision, but we can aim for integer or limited rational scale factors. These approaches still work well for quickly obtaining decent solutions that then serve as a basis for further optimization.
Let $S$ be the source graph and $D$ be the target graph.
While we do not care about the direction of edges, 
we assume both graphs to be directed (any fixed edge orientation),
simplifying the encoding.

{\scriptsize
\begin{lstlisting}
% Input: 
% vertex(x)  for every source vertex x in S
% edge(x,y)  for every source edge (x,y) in S
% tvertex(x) for every target vertex x in D
% tedge(x,y) for every target edge (x,y) in D

% Guess single scale factor
{scale(S): S >= 1, S <= |$\mid E(D) \mid$|}=1
% Minimize scale factor
|$\#$|minimize{S:scale(S)}

% Source vertices X assigned to single target A
{sass(X,A):tvertex(A)}=1 |$\leftarrow$| vertex(X)

% Reachability via source assignment
reach(e(X,Y),A) |$\leftarrow$| sass(X,A), edge(X,Y)
% Reachability via outgoing edge
reach(e(X,Y),B) |$\leftarrow$| eass(e(X,Y),_,B) 

% Guess outgoing edge among reached vertices
{eass(e(X,Y),A,B)} |$\leftarrow$| reach(e(X,Y),A), edge(A,B)
{eass(e(X,Y),A,B)} |$\leftarrow$| reach(e(X,Y),A), edge(B,A)

% Target path length must not exceed scale
|$\bot \leftarrow$| #count{A,B:eass(e(X,Y),A,B)} > S, scale(S), edge(X,Y)

% Every target edge must be covered
|$\bot \leftarrow$| tedge(A,B), not eass(_,A,B), not eass(_,B,A)

% Every source edge must be constructed
con(X,Y) |$\leftarrow$| edge(X,Y), reach(e(X,Y),B), sass(Y,B)
|$\bot \leftarrow$| edge(X,Y), not con(X,Y)

% Degree rule to ensure income equals outcome 
|$\bot \leftarrow$| O=#count{B: eass(e(X,Y),A,B)}, I=#count{B: eass(e(X,Y),B,A)}, O|$\neq$|I, reach(e(X,Y),A), not sass(Y,A), not sass(X,A)
% Source gets additional outcome
|$\bot \leftarrow$| O=#count{B: eass(e(X,Y),A,B)}, I=#count{B: eass(e(X,Y),B,A)}, O|$\neq$|I+1, reach(e(X,Y),A), sass(X,A), not sass(Y,A)
% Target gets additional income
|$\bot \leftarrow$| O=#count{B: eass(e(X,Y),A,B)}, I=#count{B: eass(e(X,Y),B,A)}, O+1|$\neq$|I, reach(e(X,Y),A), sass(Y,A), not sass(X,A)
\end{lstlisting}
}

Note that we extend this encoding to capture limited rational scale factors. However, it works similarly and was skipped for the sake of presentation. There, we subdivide target edges (which then still allows us to map source vertices only to target vertices) and consider these subdivisions when computing the scale factor.

\medskip
\subsubsection{Rational Scales via ILP}\label{sec:milp}

\noindent In the following we discuss an \emph{integer linear program (ILP)}.
To this end let $S$ be the source graph and $D$ be the target graph. Further, we assume $D'$ to be a modification of target $D$, where each edge $e=\{x,y\}$ in $E(D)$ is subdivided using up to $c\leq |V(S)|$ intermediate \emph{auxiliary vertices}. These auxiliary vertices allow us to position source vertices on edge positions and are positioned between 0 ($x$) and 1 ($y$). For every target edge $e\in E(D)$, we assume a predecessor relation $\prec$ among the vertices in $V(D')$ on $e$ obtained after subdividing. Without loss of generality, we assume that edges in $E(D)$ are directed, based on this ordering, i.e.\ edges go from $\prec$-lower to $\prec$-higher vertices.
Let $\alpha\leq |E(D)|$ be a chosen scale factor upper bound constant.

\smallskip
In the encoding, we use the following variables:

{\footnotesize
\begin{align*}
	&\text{directed flow} &&f_{s,x,y} \geq 0 && \text{for } s \in E(S), \{x,y\}\in E(D')\\
    &\text{directed edge decision} && e_{s,x,y}: \text{binary} \\ %
	&\text{src mapping} && m_{u,n}: \text{binary} && \text{for } u \in V(S), n\in V(D')\\
	&\text{aux position} && a_{n}: [0,1] && \text{for } n \in V(D')\\
	&\text{src reachability} && r_{s,n}: \text{binary} && \text{for } s\in E(S), n\in V(D')\\
	&\text{directed edge length} && l_{s,x,y}: [0,1] && \text{for } s\in E(S), \{x,y\}\in E(D')\\
	&\text{objective} && o \geq 0 && %
\end{align*}}

\smallskip
For every $s\in E(S)$, we then generate the following ILP: %

{\footnotesize
\begin{align*}
	&\text{edge/flow rel.} && e_{s,x,y} \leq f_{s,x,y} && \text{for }%
	\{x,y\}\in E(D')\\
	&\text{} && f_{s,x,y} \leq \alpha \cdot e_{s,x,y}\\ %
	&\text{edge/reach rel.} && e_{s,n,x} \leq r_{s,n} && \text{for }%
	\{n,x\}{\,\in\,}E(D'), n{\,\in\,}V(D') \\
	&\text{} && r_{s,n} \leq \sum_{x\in ngbs(n)}(e_{s,x,n} + e_{s,n,x})\hspace{-5em} && \text{for }%
	n\in V(D')\\
	&\text{edge/length rel.} && l_{s,x,y} \leq e_{s,x,y} && \text{for }%
	\{x,y\}\in E(D') \\
	&\text{} && l_{s,x,y} \leq a_y - a_x %
	&& \text{for }%
	x,y{\,\in\,}V(D'), x{\,\prec\,}y\\ %
	&\text{} && l_{s,y,x} \leq a_y - a_x,\quad a_y \geq a_x\\
	&\text{} && l_{s,x,y} - a_y + a_x \geq e_{s, x,y} - 1\hspace{-5em} %
	\\
	&\text{} && l_{s,y,x} - a_y + a_x \geq e_{s, y,x} - 1\hspace{-5em}\\ %
	&\text{} && a_x=0, a_y=1 && \text{for } (x,y)\in E(D)\\
	&\geq\text{demand }1 && \Sigma_{x\in ngbs(n)} f_{s,x,n} -  \Sigma_{x\in ngbs(n)} f_{s,n,x} +\hspace{-5em} %
	&& \text{for }%
	s{=}(u,v), n{\,\in\,}V(D')\\[-.5em]
	&\text{} && \qquad\qquad\qquad\quad\alpha\cdot m_{u,n} \geq r_{s,n}\\ %
	&\text{degree constr.} && \Sigma_{x\in ngbs(n)} e_{s,n,x} -  \Sigma_{x\in ngbs(n)} e_{s,x,n} -\hspace{-5em} %
	\\[-.5em]
	&\text{} && m_{u,n} + m_{v,n} = 0 \\ %
	&1\text{ target per src} && \Sigma_{n\in V(D')} m_{u,n} =1 && \text{for } u\in V(S)\\
	&\text{edges covered} && \Sigma_{s'\in E(S)} e_{s',y,x}+e_{s',x,y} \geq 1 && \text{for } \{x,y\}\in E(D')\\
	&\text{lengths }\leq o && \Sigma_{\{x,y\}\in E(D')}l_{s,x,y} \leq o\\ %
	&&&\min(o)
\end{align*}}

\medskip
\subsubsection{Local Improvement of Limited Rational Scales}\label{sec:local}

\noindent For computing and improving upper bounds among platonic solids (as presented in Table~\ref{tab:UpperBounds} works as follows.
First, using the approach discussed in Section~\ref{sec:lp}, we obtain initial mappings with integer or limited rational scale factors.
These mappings are then relaxed, where we allow a shift of source vertices mapped onto target vertices or edges.
The goal is then to use the approach of Section~\ref{sec:milp} to optimize those shifts such that the resulting solution is better or at
least only slightly worse.
We then iteratively try to keep parts of a solution that is then optimized and completed with the presented encoding.
Whenever we obtain a new solution we keep parts of it for the next iteration, assuming the quality improved or only slightly worsened.

\section{Lower Bounds}
\label{sec:lower}

\noindent For any pair of polyhedra, a na\"ive lower bound on the optimal scale factor is immediate: every target edge needs to be covered by a source edge, thus
$
\text{OPT} \geq \tfrac{\text{perimeter of target graph}}{\text{perimeter of source graph}}
$. 
We significantly improve this lower bound with the following observations: 
\begin{itemize}
    \item Each source vertex is mapped to either a vertex or a point along an edge of the target;
    \item Each source edge is routed to a path in the target;  
    \item The scale factor is the maximum length of the routed target paths.
\end{itemize}

\noindent Let $n_s$ be the number of vertices in the source graph, and let $o_t$ be the number of vertices of odd degree in the target. 
The following result forms the basis of our lower bounds: 

\medskip
 \begin{lemma}\label{lem:double}
     Every target edge not containing a source vertex is either covered (1) by one or more paths going straight along the complete edge, or (2) by two doubling-back paths that meet at a point (doubling the entire target edge).
 \end{lemma}
 \begin{proof}
If there is at least one path going straight along the complete edge, then it is never helpful to have a doubling-back path: such a doubling-back can be shortcut to just never enter the edge.
 \end{proof}

We obtain the following based on doubly covering edges. 

\medskip
\begin{lemma}\label{lem:oddeven}
    In any solution, at least $\frac{o_t-n_s}{2}$ target edges must be fully doubly covered. %
    If at least one source vertex is placed in the middle of a target edge, this bound is strict.
\end{lemma}
\begin{proof}
	Locally at an odd-degree vertex $v$ of the target graph, if $v$ does not have a vertex of the source graph at it, then $v$ is visited by a sequence of paths that enter and leave.  Because $v$ has odd degree, there must be a local doubling along an edge $e = \{v,w\}$ incident to $v$.  If this doubling continues fully along the edge $e$ then we get one doubled edge for two (potentially) odd-degree vertices, matching the bound of $o_t/2$.

	Now suppose the doubling does not continue fully along the edge $e$. If there are no source vertices along edge $e$ then by Lemma~\ref{lem:double}, the edge has two meeting doubling paths, which doubles the entire edge, and again we are done. So suppose there is a source vertex $s$ along edge $e$. If there are at least two outgoing paths in each direction from $s$, then again, the edge is fully doubled, so $s$ provided no benefit. Otherwise, there is at most one outgoing path in some direction, say $d \in \{v,w\}$. Then $d$ locally acts the same as if $s$ were not there, so it must have an incident doubling not on the edge $e$ (i.e., not toward $s$ so the argument above applies. The other vertex $\{v,w\}-\{d\}$ no longer necessarily has an incident wholly doubled edge, but this matches the ``$-n_s$'' in the bound.
Note: if there are two source vertices along the same edge $e$ then both endpoints potentially have no incident wholly doubled edge. But this again matches the ``$-n_s$'' in the bound (there are two source vertices to charge to).

	\medskip
If at least one source vertex is in the middle of a target edge, then the above bound becomes strict.
In this case, we get a doubled partial edge beyond the bound (on the side $\{v,w\}-\{d\}$ so we have strictly more doubling and get a strictly large bound.
\end{proof}

This yields the following adapted lower bound.

\medskip
\begin{theorem}[Improved Lower Bound]~\begin{center}\rm
$\text{OPT} \geq \tfrac{\text{perimeter of target graph}~+~\text{lengths of doubly covered target edges}}{\text{perimeter of source graph}}.$
\end{center}
\end{theorem}
\medskip

\begin{table*}[t]
\caption{Detailed computation of lower bounds as shown in Table~\ref{tab:UpperBounds}, providing source vertices $n_s$, target odd-degree vertices $o_t$, as well as the perimeter of source and target graph. The lower bound marked by an asterisk is subsumed by Theorem~\ref{thm:icosacube}.}
\arxiv{\footnotesize}
\begin{tabularx}{\textwidth}{l@{\hspace{.25em}}l@{\hspace{.25em}}|@{\hspace{.45em}}X@{\hspace{.45em}}X@{\hspace{.8em}}X@{\hspace{.8em}}X@{\hspace{1.3em}}X}

\toprule 
        $\Rsh$     && Tetrahedron\newline ($o_t=4$, perim$_t=6$) & Cube\newline ($o_t=8$, perim$_t=12$)                    & Octahedron\newline ($o_t=0$, perim$_t=12$) & Dodecahedron\newline ($o_t{\,=\,}20$, perim$_t{\,=\,}30$) & Icosahedron\newline ($o_t{\,=\,}12$, perim$_t{\,=\,}30$) \\ \midrule 
Tetrahedron\newline &($n_s=4$, perim$_s=6$) & 1  & $2+1/3=\frac{12+(8-4)/2}{6}$ & $>2=\frac{12}{6}$  & $6+1/3=\frac{30+(20-4)/2}{6}$  & $>5+2/3{=}\frac{30+(12-4)/2}{6}$       \\[.25em] %
Cube &($n_s=8$, perim$_s=12$) & $>1/2=\frac{6+(4-4)/2}{12}$ & 1 & $>1=\frac{12}{12}$    & $3=\frac{30+(20-8)/2}{12}$   & $>2+2/3{=}\frac{30+(12-8)/2}{12}$   \\[.25em] %
Octahedron &($n_s=6$, perim$_s=12$)  & $>1/2=\frac{6}{12}$ & ${>\,}1+1/12{=}\frac{12+(8-6)/2}{12}$ & 1   & ${>\,}3{\,+\,}1/12{=}\frac{30+(20-6)/2}{12}$  & $>2+1/2{=}\frac{30+(12-6)/2}{12}$  \\[.25em] %
Dodecahedron &($n_s=20$, perim$_s=30$) & $>1/5=\frac{6}{30}$ & $>2/5=\frac{12}{30}$ & $>2/5=\frac{12}{30}$ & 1 & $>1=\frac{30}{30}$  \\[.25em] %
Icosahedron &($n_s=12$, perim$_s=30$) & $>1/5=\frac{6}{30}$  & $>2/5=\frac{12}{30}\ast$  & $>2/5=\frac{12}{30}$  & ${>\,}1{\,+\,}2/15{=}\frac{30+(20-12)/2}{30}$  & 1 \\ \bottomrule                             
\end{tabularx}
\label{tab:lowerbounds}
\end{table*}

Table~\ref{tab:lowerbounds} details computation steps to obtain the lower bounds given in Table~\ref{tab:UpperBounds}. Bounds are strict according to Lemma~\ref{lem:oddeven} and by the observation that mapping between vertices of different degrees (or from larger to smaller degree vertices) causes doubly covered edge parts. %
If mapping from icosahedron to cube, we obtain the following stronger lower bound.

\begin{figure}
    \centering
\includegraphics[width=\arxiv{.75}\columnwidth]{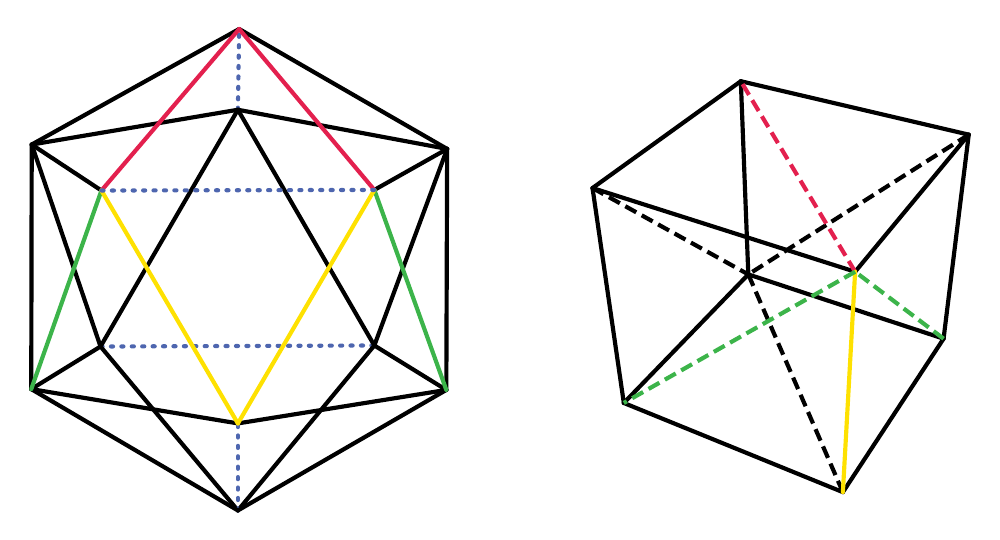}
    \caption{(Left): Icosahedron, where the four edges (dotted) highlighted in blue, are contracted. (Right): Resulting cube with additional diagonal edges (dashed). Two vertices are of degree $6$ and the remaining $6$ vertices have degree $4$.}
    \label{fig:icosa_cube}
\end{figure}

\smallskip
\begin{theorem}\label{thm:icosacube}
There cannot exist an isometric covering from icosahedron to cube.
\end{theorem}
\begin{proof}
By contracting four edges (highlighted in color in Figure~\ref{fig:icosa_cube} (left)), we obtain a cube (right) with additional diagonals. By contracting edges, we simplified the source linkage by assigning certain source vertices to the same position. Therefore, lower bounds are preserved using the simplified source. %
Note that a cube without diagonals can never provide the same connectivity as a cube. Therefore, we need additional scaling, as certain edges of the source linkage can never be represented in the target linkage via an isometric covering.
\end{proof}
\section{Hardness and Inapproximability}\label{sec:hard}

\noindent We show that deciding if one embedded (planar) graph can be mapped onto another via folding is NP-complete. Further deciding the scale factor needed to allow one planar graph to be mapped onto another cannot be approximated within a factor of $1.5-\epsilon$ for any constant~$\epsilon$. %

\begin{figure}
    \centering
    \includegraphics[width=\arxiv{.75}\journal{0.99}\linewidth]{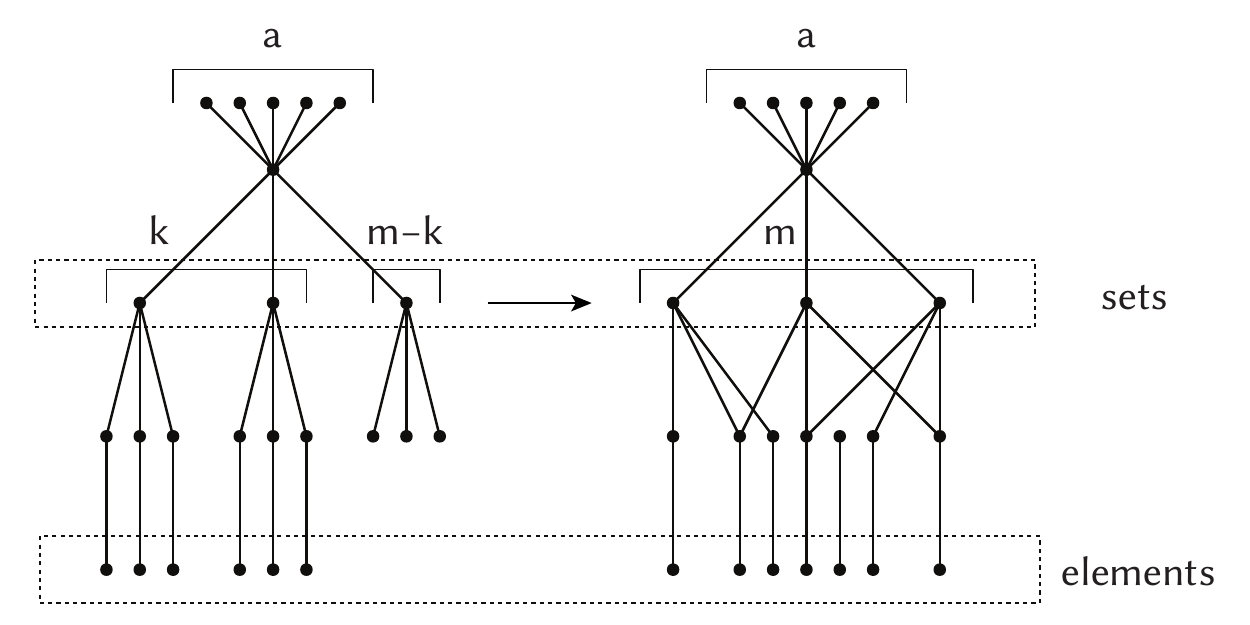}
    \caption{Proof sketch of our reduction from \emph{Planar Set Cover} to \emph{IsoCovering}.}
    \label{fig:proof}
\end{figure}

\medskip
\begin{theorem}[Inapproximability]\label{thm:inapprox}
Given two metric planar graphs $G_1$, $G_2$ and scale factor~$\alpha$, deciding whether $G_1$ can be folded onto $G_2$ via an isometric covering (scale factor~$1$) is NP-complete. Further, the optimal scale factor $\text{OPT}$ of mapping $G_1$ onto $G_2$ is NP-hard to approximate within a factor of $< 1.5$.
\end{theorem}
\begin{proof}
Membership immediately follows from utilizing an ILP that is presented in Section~\ref{sec:milp}.

For hardness, we reduce from \emph{Planar Set Cover} with sets of size exactly 3~\cite{lichtenstein1982planar}%
. Here we are given a set of elements $X$ and a collection $C$ of $m$ sets of three elements of $X$. Further, if one constructs the bipartite graph with notes for elements of $C$ and $X$ and edges when $x\in c$, then that graph is planar. The question is whether there is a subset $C' \subset C$ of size $k$ such that each element of $X$ is contained in $C'$.

In this construction, which is sketched in Figure~\ref{fig:proof}, all edges are of length 1. For the target graph first construct a top node and one node for each set $c$ and connect them to the root. Also connect $a>4|C|$ nodes to the root, which will help anchor the root node in the construction. Then, we construct a node for each element $x$ and connect it to the nodes representing sets containing that element. Finally, off of each $x$, add an additional degree-1 node connected to it. This has a total of $|C|+2|X|+1+a$ nodes and $4|C|+|X|+a$ edges.

The initial graph is a tree. The root node has $a+|C|$ children. Of these $a$ children will be leaf nodes. Then $k$ children will each be connected to three length 2 paths which will be the intended to be used to solve the set cover of size $k$. The other $|C|-k$ children will be connected to $3$ leaves and be used to cover the edges in the unused sets. This graph has $4|C|+k+1+a$ nodes and $4|C|+k+a$ edges.

In this construction, due to the $a$ leave nodes, a valid solution must map the root node to the top node. 

If there is a valid solution to the set cover, take the $k$ subtrees of depth 3 and map them to the sets in the set cover solution. This will allow the $k$ paths of length 2 to cover the element nodes and their leaves. Next, use the remaining subtrees to cover all the unchosen set nodes and the edges of their elements.

If a solution to the set cover does not exist we will show there is no solution to the graph folding problem. First, note that all of the element leaves in the target are distance three from the top node and that these leaves are distance 4 from each other. Thus we know the leaves of the $k$ subtrees of depth 3 must be used to cover the leaves of the element nodes. These subtrees are distance 1 from the root and thus must either be mapped to the $a$ nodes (which then prevents them from covering element leaves) or to set nodes. The element leaves in a set are of distance 2 from the node containing them and thus can be covered by the leaves of a subtree mapped to that set node. However, all leaves of the elements not contained in that set are of distance at least 4 and thus cannot be covered. Thus, if there is no set cover of size $k$, some of the element leaves will be impossible to cover in the graph folding.
\end{proof}

For establishing also membership in NP, we rely on the insights from the next section.

\section{The Optimum is Rational}\label{sec:rational}
\noindent It turns out that the best scale factor is guaranteed to be rational,
which we prove as follows.

\smallskip
\begin{theorem}\label{thm:rational}
Given graphs $G_1$ and $G_2$, the optimal scale $\alpha$ such that $G_1$ can be folded onto $G_2$ via isometric covering is rational.
\end{theorem}
\begin{proof}
	We will prove the result by designing an ILP, where we can compute all (exponential) possibilities for binary variable assignments. After assigning these integer variables, we invoke the remaining LP. This LP only has rational coefficients, so the optimal solution will be rational~\cite{AspvallStone80,Megiddo83,stackexchange}. Consequently, optimal scale factors are rational.	

	\medskip
	Recall the ILP encoding from Section~\ref{sec:milp}, as well as scale factor upper bound~$\alpha$, source graph~$S$, target graph~$D$ and $D'$ that is obtained from $D$ after subdividing. After fixing integer variables, the following variables remain: %
	
	{\footnotesize
\begin{align*}
	&\text{directed flow} &&f_{s,x,y} \geq 0 && \text{for every } s \in E(S), \{x,y\}\in E(D')\\
	&\text{aux position} && a_{n}: [0,1] && \text{for every } n \in V(D')\\
	&\text{directed edge length} && l_{s,x,y}: [0,1] && \text{for every } s\in E(S), \{x,y\}\in E(D')\\
	&\text{objective} && o \geq 0 && %
\end{align*}}

	\medskip
	Then, if some of the constraints containing integer variables are not satisfied, we can stop.
	For the remainder, we are left with the following non-integer constraints, constructed for every~$s\in E(S)$, as the remaining integer variables are fixed constants:
	
	{\footnotesize
\begin{align*}
	&\text{edge/length rel.}\hspace{-.75em} && l_{s,x,y} \leq e_{s,x,y} && \text{for }%
	\{x,y\}\in E(D') \\
	&\text{} && l_{s,x,y} \leq a_y - a_x\quad l_{s,y,x} \leq a_y - a_x\hspace{-.75em} && \text{for }%
	x,y\in V(D'), x\prec y\\ %
	&\text{} && a_y \geq a_x \\ %
	&\text{} && l_{s,x,y} - a_y + a_x \geq e_{s, x,y} - 1 %
	\\ %
	&\text{} && l_{s,y,x} - a_y + a_x \geq e_{s, y,x} - 1\\
	&\text{} && a_x=0, a_y=1 && \text{for } (x,y)\in E(D)\\ %
	&\geq\text{demand }1 && \Sigma_{x\in ngbs(n)} f_{s,x,n} -  \Sigma_{x\in ngbs(n)} f_{s,n,x} +\hspace{-.75em} %
	&& \text{for }%
	s=(u,\cdot), n\in V(D')\\[-.5em]
	&\text{} && \qquad\qquad\qquad\quad\alpha\cdot m_{u,n} \geq r_{s,n}\hspace{-.75em} &&\\ %
	&\text{lengths }\leq o && \Sigma_{\{x,y\}\in E(D')}l_{s,x,y} \leq o\\
	&&&\min(o)
\end{align*}}%
The LP,
which can be solved in polynomial time~\cite{AspvallStone80,Megiddo83}, correctly characterizes the optimal scale factor $o$.
\end{proof}

\subsection{Complexity Consequences}\label{sec:consequences}
\noindent This insight results in NP membership for verifying a given scale factor. %

\medskip
\begin{corollary}\label{cor:np}
Given two graphs $G_1$, $G_2$, and a scale factor~$\alpha$, deciding whether $G_1$ scaled by $\alpha$ can be folded onto $G_2$ via an isometric covering %
is in NP. 
\end{corollary}
\begin{proof}
	The NP algorithm is given in the proof of Theorem~\ref{thm:rational}.
\end{proof}

Therefore we can use binary search with logarithmic many calls to an NP oracle to compute the optimum. %
Verifying the optimality of a given scale factor~$\alpha$ is complete for complexity class
$\text{DP} = \{a \cap b \mid a \in \text{NP}, b \in \text{coNP}\}$.

\medskip
\begin{theorem}
Given graphs $G_1$, $G_2$, and a scale factor~$\alpha$, it is DP-complete to decide whether $\alpha$ is the smallest scale factor to enable isometric covering of folding $G_1$ onto~$G_2$.
\end{theorem}
\begin{proof}
 Verifying that $\alpha$ is a valid scale factor is in NP by %
 Corollary~\ref{cor:np}.
 Deciding whether there is a strictly smaller scale factor $\alpha'<\alpha$ is in NP, so verifying the nonexistence of such a scale factor is in coNP.
 Because both can be verified independently, the problem is in DP.

 Hardness follows by applying the construction of Theorem~\ref{thm:inapprox} two times. We reduce from a tuple $(P,P')$ comprising the NP-hard problem \emph{Planar Set Cover}~$P$ and the coNP-hard coproblem \emph{Planar Set Uncover}~$P'$. For $P'$, we use the same approach as in Figure~\ref{fig:proof}, where we have a scale factor~$\alpha'\geq 1.5$ in case $P'$ is a no-instance. For $P$, we still use this approach, but with a different large-degree $2a$ (to distinguish $P$ from $P'$) and where we replace every edge of the source tree by a path comprising $3$ edges and every edge of the target tree by a path of $4$ edges.
 
 Consequently, the ideal scale factor~$\alpha$ for $P$ (if it is a yes instance) would be $\alpha=\frac{4}{3} \approx 1.333$ instead of $1$. Otherwise (in case of a no-instance), the scale would be $1.5 \cdot \frac{4}{3}=2$, which is, however, larger than $1.5$. As a result, if we manage a combined scale factor of $\max\{\alpha, \alpha'\}=1.5$, we know that $P$ is a yes-instance (as the factor is smaller than $2$). Further, if we do not manage combined scale factor $<1.5$, $P'$ is a no-instance (as desired). %
\end{proof}

\section{Open Problems}\label{sec:open}

The next steps include tightening the bounds in Table~\ref{tab:UpperBounds} and developing heuristics for folding general polyhedra. 
We would also like to strengthen the hardness proof to apply to polyhedral graphs, which requires more connectivity than our current construction. 
Finally, it would be interesting to develop solutions that permit \textit{continuous} folding motions, which are required in the case of \textit{rigid} polyhedral linkages. 

\bibliography{main}
\journal{
\begin{biography}
\profile[photos/lily]{Lily Chung}{is a Ph.D student at MIT CSAIL working with Erik Demaine and Ronitt Rubinfeld.  She received her B.Sc. in Computer Science from MIT.  She works on sublinear algorithms, computational geometry, and algorithmic lower bounds.}

\profile[photos/edemaine]{Erik D. Demaine}{received a B.Sc. degree from Dalhousie University in 1995,
and M.Math. and Ph.D. degrees from University of Waterloo in 1996 and 2001, respectively. Since 2001, he has been a
professor in computer science at the Massachusetts Institute of Technology. His
research interests range throughout algorithms, from data structures for improving web searches to the geometry of understanding how proteins fold to the computational
difficulty of playing games. In 2003, he received a MacArthur Fellowship as a ``computational geometer tackling and solving difficult problems related to folding and bending—moving readily
between the theoretical and the playful, with a keen eye to revealing the former in the latter''. He cowrote a book about the theory of folding, together with Joseph O’Rourke (\textit{Geometric Folding
Algorithms}, 2007), and a book about the computational complexity of games, together with Robert Hearn (\textit{Games, Puzzles, and
Computation}, 2009).}
\profile[photos/mdemaine]{Martin L. Demaine}{is an artist and mathematician. He started the first private
hot glass studio in Canada and has been
called the father of Canadian glass. Since
2005, he has been the Angelika and Barton Weller Artist-in-Residence at the Massachusetts Institute of Technology. Both
Martin and Erik work together in paper,
glass, and other material. They use their exploration in sculpture
to help visualize and understand unsolved problems in mathematics, and their scientific abilities to inspire new art forms. Their
artistic work includes curved origami sculptures in the permanent
collections of the Museum of Modern Art (MoMA) in New York,
and the Renwick Gallery in the Smithsonian. Their scientific
work includes over 100 published joint papers, including several
about combining mathematics and art.}
\profile[photos/hecher_photo]{Markus Hecher}{is a postdoctoral researcher at MIT CSAIL working with Erik Demaine. He received a binational PhD in computer science from the Vienna University of Technology (Austria) and the University of Potsdam (Germany). His main interests are counting problems, computational complexity, and utilizing new insights into the hardness of combinatorial problems to improve existing algorithms.}
\profile[photos/rebecca]{Rebecca Lin}{is a Ph.D. student at MIT CSAIL working with Erik Demaine. She received her B.Sc. in Computer Science from the University of British Columbia, where she was advised by William Evans. Her research explores geometrical problems in art, design, and fabrication.}%

\profile[photos/jayson-lynch-p-500]{Jayson Lynch}{is a Research Scientist in the FutureTech Lab at MIT working on predicting the future progress of algorithms development and understanding the fundamental limitations to improving computing performance. They earned their PhD from MIT under Erik Demaine working on the computational complexity of motion planning problems, computational geometry, and games and puzzles.}
\profile[photos/selfphto_nara]{Chie Nara}{received her B.A., M.S., and Ph.D. degrees from Ochanomizu University in Tokyo.
She served as a professor at Tokai University before taking up the current research position at Meiji University. Her
research fields are functional analysis, graph theory, and discrete geometry.}
\end{biography}}

\end{document}